\documentclass[12pt]{article}
\usepackage{amsmath,amssymb,amsfonts,amsthm}
\usepackage{pdfsync}
\usepackage{graphicx}
\usepackage{amssymb}
\usepackage{epstopdf}
\usepackage{amsfonts}
\usepackage{mathrsfs}
\usepackage{marvosym}
\usepackage[pdftex]{color}
\usepackage{pdfcolmk}

\usepackage[dvips]{epsfig} \DeclareGraphicsExtensions{.ps,.eps}
\def\R{\mathbb{R}}

\newcommand{\tp}{^{\top}}

\newcommand{\beq}{\begin{equation}}
\newcommand{\eeq}{\end{equation}}
\newcommand{\bea}{\begin{eqnarray}}
\newcommand{\eea}{\end{eqnarray}}

\newcommand{\bsea}{\begin{subeqnarray}}
\newcommand{\esea}{\end{subeqnarray}}
\newcommand{\nn}{\nonumber}

\def\bmat{\left[ \begin{array}}
\def\emat{\end{array} \right]}

\newcounter{acount}

\newtheorem{theorem}{Theorem}

\newtheorem{definition}{Definition}[section]
\newtheorem{lemma}{Lemma}[section]

\newtheorem{conjecture}{Conjecture}[section]

\title{On the Lyapunov equation with the state matrix in companion form}
\author{Augusto Ferrante\thanks{Augusto Ferrante is with University of Padova, Dept. of Information Engineering, Via Gradenigo 6/a, 35131 Padova, Italy. {\tt\small augusto@dei.unipd.it}}}
\begin{document}
\maketitle

\begin{abstract}
We study the continuous-time Lyapunov equation
under the assumption that the state matrix is a Hurwitz companion matrix.  The standard
Lyapunov theory implies that the unique solution $X$ is positive
semidefinite.  Motivated by positive systems, we investigate the 
 question of whether $X$ is entrywise nonnegative.  We prove that
this is the case when the companion matrix has only real eigenvalues.  The
proof reduces each entry of $X$ to a quadratic form associated with a class of
Cauchy-like matrices whose entries are expressed in terms of elementary
symmetric polynomials.  The required nonnegativity then follows from the
positive semidefiniteness of these Cauchy-like matrices.  We also discuss a
stronger total-positivity property: total nonnegativity does not hold in
general, but it is recovered under an additional sign condition on the
expansion of the forcing vector in the eigenbasis of $A\tp$.
\end{abstract}

\section{Introduction}
This paper is dedicated to Prof. Luigi Fortuna, now retired from the University of Catania, as a token of friendship and esteem. 
More than $15$ years ago, Prof. Fortuna presented a conjecture 
about Lyapunov equations \cite{gigi-pc}. This paper addresses at least part of this conjecture
after a long period of thinking.

Let 
\beq\label{carpol}
a(s)=s^n+a_{n-1}s^{n-1}+a_{n-2}s^{n-2}+\dots+ a_{1}s+a_0
\eeq
and assume that $a(s)$ is a Hurwitz polynomial (so that, in particular,
$a_i>0$, for all $i=0,1,\dots, n-1$).

Let 
$A\in\R^{n\times n}$ be the companion matrix associated with $a(s)$:
\beq\label{compA}
A:=\bmat{cccccc}
0&1 &0 &\dots &0\\
0&0 &1 &\ddots &0\\
\vdots&\ddots&\ddots&\ddots&\ddots\\
0&0&0&\dots&1\\
-a_0&-a_1&-a_2&\dots&-a_{n-1}
\emat.
\eeq
Let $Q=Q\tp\in\R^{n\times n}\succeq 0.$
Let $X$ be the only solution of the Lyapunov equation
\beq\label{lyap-gen}
XA+A\tp X=-Q.
\eeq
Notice that by the standard theory on Lyapunov equations $X=X\tp$ is positive semidefinite.

\begin{conjecture}\label{conj}
Let $X$ be the only solution of the Lyapunov equation (\ref{lyap-gen}), where $A$ has the form 
(\ref{compA}), the polynomial $a(s)$ given by (\ref{carpol}) is Hurwitz and $Q=Q\tp\succeq  0$.
Then $X$ is entrywise nonnegative.
\end{conjecture}

It is useful to distinguish two different positivity properties of the
solution of a Lyapunov equation.  The classical one is positivity in the
Loewner order (the partial order defined by the convex cone of positive semi-definite matrices): if $A$ is Hurwitz and $Q=Q\tp\succeq0$, then the unique
solution of (\ref{lyap-gen})
is symmetric and positive semidefinite.  This is the positivity property that
is normally used in Lyapunov stability theory, where $x\mapsto x\tp Xx$ is a
quadratic storage or Lyapunov function.

The question considered here concerns a different order, namely the
componentwise order.  We ask whether the same solution $X$ is entrywise
nonnegative, that is,
\[
X_{ij}\ge0,\qquad i,j=1,\dots,n.
\]
This property is not a  consequence of positive semidefiniteness: a
positive semidefinite matrix may have negative off-diagonal entries.  Thus the
entrywise nonnegativity of $X$ is an additional structural property of the
Lyapunov equation, depending in an essential way on the companion form of the
state matrix.

Entrywise positivity is especially natural in the theory of positive systems.
A continuous-time linear system
\[
\dot x=Ax+Bu,\qquad y=Cx+Du,
\]
is called {\em positive} if nonnegative initial conditions and nonnegative inputs
produce nonnegative states and outputs for all future times.  Equivalently, in
the finite-dimensional continuous-time case, $A$ is a Metzler matrix and
$B,C,D$ are entrywise nonnegative.  Positive systems occur in many models in
which the state variables represent intrinsically nonnegative quantities, such
as masses, concentrations, populations, probabilities, queues, or economic
stocks.  Typical application areas include compartmental systems,
pharmacokinetics, population dynamics, epidemiology, chemical reaction
networks, ecology, economics, and Markov or queueing models; see, for example,
\cite{Luenberger1979,FarinaRinaldi2000,BermanPlemmons1994}.
For such systems, entrywise nonnegativity of matrices arising from the system
has a direct interpretation: it says that the corresponding map preserves the
componentwise order.  In this sense, while positive semidefiniteness is the
standard Lyapunov-theoretic positivity property, entrywise nonnegativity is an
additional positivity property motivated by positive systems.  The purpose of
this paper is to prove that this additional property holds for the Lyapunov
equation above when the state matrix is in companion form and has only real
eigenvalues. The main motivation of Prof. L. Fortuna for addressing this problem is connected to the regulator where the optimal gain can be expressed in terms of the solution of an algebraic Riccati equation (ARE) which can be rewritten as a Lyapunov equation. For the single input case, when the input matrix is entrywise nonnegative, if the solution of the ARE is also entrywise  nonnegative
the optimal regulator can be implemented by using amplifiers with only positive 
voltage in the power supply which is a great hardware advantage.

The main result of this paper is to prove Conjecture \ref{conj} at least in a special case.

\begin{theorem}\label{main-res}
Under the assumptions of Conjecture 
\ref{conj}, if $A$ has only real 
eigenvalues then $X$ is entrywise nonnegative.
\end{theorem}

\begin{remark}
Since $X$ is clearly positive semidefinite, it is obvious that the diagonal elements of $X$ are nonnegative.
Moreover, it is not difficult to show by direct computation that the elements $X_{1,n}$ and $X_{n,1}$ namely the entries of $X$ in the top-right and in the bottom-left corners are also nonnegative.
Finally, the cases $n=2,3,4$ can be worked out by brute force and, in these cases it is possible to see that the conjecture holds.
\end{remark}

\section{Problem simplification}

In this section, we show that we can make two useful assumptions without loss of generality.
This holds equally for the original conjecture and for our main result. 

First of all notice that\footnote{This first simplification was first suggested to me by my dear friend, Prof. Harald Wimmer, \cite{harald-pc}.}, since any positive semidefinite matrix $Q=Q\tp$ can be written as the sum of rank 1 positive semidefinite  matrices, i.e.
\beq\label{qsum}
Q=\sum_{i=1}^k q_iq_i\tp,\quad q_i\in\R^n,
\eeq
we can assume  without loss of
generality that $Q=qq\tp,$ with $q\in\R^n$. 
In fact, by linearity, if the solution $X$ of (\ref{lyap-gen}) is entrywise non-negative for all
$Q=qq\tp,$ then it is also  entrywise non-negative for any $Q$ of the form (\ref{qsum}).

Second, we can assume without loss of generality that $A$ given by (\ref{compA}) is diagonalizable.
Indeed, as long as the spectrum of $A$ is in the open left-half plane the solution $X$ 
of (\ref{lyap-gen}) is a continuous function of $A$ and of $Q$. Assume,
that Conjecture \ref{conj} holds whenever $A$ is diagonalizable and suppose, by way of contradiction, that there exist $Q=Q\tp\succeq  0$ and a non-diagonalizable  $A_0$ such that, for $A=A_0$, at least one entry  of the solution $X_0$ of (\ref{lyap-gen}) is negative, say  $[X_0]_{hk}<0$.
Let
$$
a_0(s)=\det(sI-A_0)=\prod_{i=1}^k (s+\lambda_i)^{\nu_i},\qquad \lambda_i\neq\lambda_j\ \forall i\neq j.
$$
We can perturb this polynomial by changing each term of the form $(s+\lambda_i)^{\nu_i}$ with
$\prod_{l=0}^{\nu_i-1}(s+\lambda_i+l\varepsilon),$ where $\varepsilon$ is a constant.
In this way we obtain a polynomial $a_\varepsilon(s)$ whose zeros are all simple, as long as
$\varepsilon>0$ is sufficiently small. Moreover, by selecting a small enough value of $\varepsilon>0$, the coefficients of $a_\varepsilon(s)$ can be made arbitrarily close to those of $a_0(s)$. Notice, in passing, that even if some of the zeros of $a_0(s)$ are not real, the coefficients of $a_\varepsilon(s)$ remain  real.
In conclusion, by selecting a small enough value of $\varepsilon>0$, the companion matrix $A_\varepsilon$ associated with the polynomial $a_\varepsilon(s)$  can be made diagonalizable and arbitrarily close to
$A_0$. By continuity, the associated solution $X_\varepsilon$ can be made arbitrarily close to
$X_0$ and hence, there exists $\varepsilon>0$ for which $[X_\varepsilon]_{hk}<0$.
This is a contradiction because $A_\varepsilon$ is diagonalizable.
Finally, notice that since any companion matrix is non-derogatory, then $A$ diagonalizable is equivalent to the fact that all its eigenvalues have algebraic multiplicity equal to $1$.

We observe that these simplifications can be made for the general problem as they do not require the assumption of $A$ having real eigenvalues.

\section{Preliminaries}

We start by recalling the definition of elementary symmetric polynomials.
\begin{definition}
Let $x_1,\dots,x_r$ be real variables and let $m$ be an integer with
$0\le m\le r$.
The \emph{elementary symmetric polynomial of degree $m$}
in the variables $x_1,\dots,x_r$ is defined by
\[
\sigma_m(x_1,\dots,x_r)
:=
\sum_{1\le i_1<i_2<\cdots<i_m\le r}
x_{i_1}x_{i_2}\cdots x_{i_m}.
\]
By convention, $\sigma_0(x_1,\dots,x_r):=1$ and $\sigma_m(x_1,\dots,x_r)=0$ for $m<0$ and for $m>r$.
\end{definition}

Next, we derive a full set of eigenvectors of $A\tp$ assuming that $A$ given by (\ref{compA}) is diagonalizable. To this aim recall that if $A$ is a diagonalizable companion matrix (equivalently, if  $A$ is a  companion matrix whose eigenvalues have all algebraic multiplicity equal to $1$) then the changes of basis matrices diagonalizing $A$ and $A\tp$ 
have Vandermonde structure.
More explicitly (see e.g.\cite{turnbull1961introduction} or the Wikipedia page on companion matrices), if $A$ is given by (\ref{compA}) and the associated polynomial $a(s)$ has $n$ distinct zeros $-\lambda_i$  
then
$$
VA\tp V^{-1} = D,
$$
where 
$$
V:=\bmat{cccccc}
1&-\lambda_1 &(-\lambda_1)^2 &\dots &(-\lambda_1)^{n-1}\\
1&-\lambda_2 &(-\lambda_2)^2 &\dots &(-\lambda_2)^{n-1}\\
1&-\lambda_3 &(-\lambda_3)^2 &\dots &(-\lambda_3)^{n-1}\\
\vdots&\vdots&\dots&\vdots&\vdots\\
1&-\lambda_n &(-\lambda_n)^2 &\dots &(-\lambda_n)^{n-1}
\emat, $$
is the Vandermonde matrix associated with the nodes $-\lambda_i$ and
$$D:= \bmat{cccccc}
-\lambda_1 &0 &\dots&\dots &0\\
0&-\lambda_2 &0 &\dots &0\\
0&0&-\lambda_3 &\ddots &0\\
\vdots&\vdots&\dots&\vdots&\vdots\\
0&\dots & \dots&\ddots &-\lambda_n
\emat
$$
is a diagonal matrix having the eigenvalues of $A$ in the main diagonal.
We now need an expression for $V^{-1}$. This is the inverse of a Vandermonde matrix 
which has the following explicit form, see e.g. \cite{EISINBERG20061384}:
denote by $\sigma_{m}:=\sigma_{m}\{\lambda_1, \dots, \lambda_n\}$
the  elementary symmetric polynomial of degree $m$
in the variables $\{\lambda_1, \dots, \lambda_n\}$ and by
$\sigma_{m}^{(j)}:=\sigma_{m}\{\lambda_1, \dots,\lambda_{i-1},\lambda_{i+1},\dots \lambda_n\}$
the  elementary symmetric polynomial of degree $m$
in the variables $\{\lambda_1, \dots, \lambda_n\} \setminus \{\lambda_j\}$. Then the elements of $V^{-1}$ are
\begin{equation}\label{eigvinv}
[V^{-1}]_{ij} = \frac{\sigma_{n-i}^{(j)}}{\prod_{\substack{k=1 \\ k \neq j}}^n (\lambda_k - \lambda_j)}.
\end{equation}
The columns of $[V^{-1}]_{ij}$ are a full set of eigenvectors of $A\tp$ corresponding to distinct eigenvalues
$-\lambda_j$. Since all the elements on each such column has the same denominator, we can eliminate such denominators and still have a set of eigenvectors of $A\tp$ corresponding to distinct eigenvalues.
In conclusion, we have the following result.
\begin{lemma}
Assume that $A$ is diagonalizable and let $-\lambda_i$, $i=1,\dots,n$ be its distinct eigenvalues.
Then
\beq\label{eigvj}
v_j:=\begin{bmatrix} \sigma_{n-1}^{(j)}\\ \vdots\\ \sigma_0^{(j)}\end{bmatrix},\qquad j=1,\dots,n
\eeq
is such that
$$
A\tp v_j=-\lambda_j v_j,\qquad j=1,\dots,n.$$
\end{lemma}

\section{Proof of the main result}

Consider the equation 
\beq
\label{lyap-red}
XA+A\tp X=-qq\tp,\quad q\in\R^n
\eeq
and assume that $A$ is diagonalizable and has eigenvalues
$-\lambda_i$, $i=1,\dots n$, with $\lambda_i\in\R$ 
and order them as follows:
$$0<\lambda_1 < \lambda_2 <\dots< \lambda_n.
$$

Since the vectors $v_j$ defined in (\ref{eigvj}) correspond to distinct eigenvalues, they are linearly independent and hence they span the whole $\R^n$.
As a consequence, there exist real numbers $a_i$, $i=1,\dots,n$, such that
$q=\sum_{i=1}^n a_i v_i$.
Set $\mathbf{a}:=[a_1\ a_2\ \dots\  a_n]^\top$.
We now use the integral formula for the solution of Lyapunov equations and get
\begin{align} \nn
X&=\int_{0}^{\infty}\exp(A\tp t)qq\tp\exp(A t)dt 
\\
\nn
&=
  \sum_{i=1}^n
  \sum_{j=1}^n
  a_ia_j
  \int_{0}^{\infty}\exp(A\tp t)v_i v_j\tp\exp(A t)dt 
  \\
  \nn
  &=
  \sum_{i=1}^n
  \sum_{j=1}^n
  a_ia_j
  \int_{0}^{\infty}\exp(-\lambda_i t)\exp(-\lambda_j t) v_i v_j\tp dt 
  \\
  &=
  \sum_{i=1}^n
  \sum_{j=1}^n
  \frac{1}{\lambda_i +\lambda_j}
  a_ia_j v_i v_j\tp 
\end{align} 

Thus, for $k,l=1,\dots,n$, the entry $X_{kl}$ is given by
\begin{align} \nn
X_{kl}&=
  \sum_{i=1}^n
  \sum_{j=1}^n
  \frac{1}{\lambda_i +\lambda_j}
  a_ia_j [v_i v_j\tp]_{kl}\\
  \nn
  &=
  \sum_{i=1}^n
  \sum_{j=1}^n
  a_i
  \frac{\sigma_{n-k}^{(i)}\sigma_{n-l}^{(j)}}{\lambda_i +\lambda_j}
  a_j\\
  \nn
  &=
  \sum_{i=1}^n
  \sum_{j=1}^n
  a_i
  \frac{\sigma_{n-k}^{(i)}\sigma_{n-l}^{(j)}}{\lambda_i +\lambda_j}
  a_j\\
  \nn
  &=
  \mathbf{a}^\top \hat{M} \mathbf{a}\\
 &=
  \frac{1}{2} \mathbf{a}^\top M \mathbf{a}
\end{align}
where $\hat{M}$ is the matrix whose elements are $\hat{M}_{ij}:=\frac{\sigma_{n-k}^{(i)}\sigma_{n-l}^{(j)}}{\lambda_i +\lambda_j}$ and $M:=\hat{M}+\hat{M}\tp$ so that 
$$M_{ij}:=\frac{\sigma_{n-k}^{(i)}\sigma_{n-l}^{(j)}+\sigma_{n-k}^{(j)}\sigma_{n-l}^{(i)}}{\lambda_i +\lambda_j}$$

Thus, we are reduced to show that for any $k,l=1,\dots,n$ the matrix $M$ just defined is positive semidefinite.
To this aim, 
let $$r_i(k,l):= \frac{\sigma_{n-k}^{(i)}}{\sigma_{n-l}^{(i)}}$$
and observe that
$$M_{ij}:=\frac{\sigma_{n-k}^{(i)}\sigma_{n-l}^{(j)}+\sigma_{n-k}^{(j)}\sigma_{n-l}^{(i)}}{\lambda_i +\lambda_j}=\sigma_{n-l}^{(j)}\sigma_{n-l}^{(i)} \left( \frac{r_i(k,l) + r_j(k,l)}{\lambda_i + \lambda_j} \right) $$
so that
$$M=\Delta C\Delta$$
with $$\Delta:={\rm diag}(\sigma_{n-l}^{(1)}, \dots, \sigma_{n-l}^{(n)})$$
and
$C$ is a Cauchy-like matrix whose elements are
\beq\label{Cauchy-likemat}
C_{ij}=\frac{r_i(k,l) + r_j(k,l)}{\lambda_i + \lambda_j}.
\eeq
Since $\Delta=\Delta\tp$ is nonsingular, $M\succeq 0$ is equivalent to $C \succeq 0$.
But the matrix $C$ is known to be positive semidefinite for all $k,l=1,\dots,n$, \cite{ferrante2026positivityclasscauchylikematrices}.
This concludes the proof.

\section{Totally positive solutions}

In this section we consider a stronger form of positivity.\footnote{Here we follow the notation of \cite{Pinkus-TPM}. Other works define {\em totally
nonnegative} the matrices defined as totally positive by \cite{Pinkus-TPM}
and totally positive the matrices that in \cite{Pinkus-TPM} are called 
strictly
totally positive.}

\begin{definition}
A matrix $P$ is said to be {\em totally positive} if
all its minors are nonnegative; it is said to be {\em strictly
totally positive} if all its minors are strictly positive.
\end{definition}

It is natural to suspect that, under our assumptions, the solution $X$
of (\ref{lyap-red}) is totally nonnegative.
This is in general not the case  as
shown by the following counterexample.
Let  
$$
A=\bmat{ccc}0&    1& 0\\
         0   &       0   &  1\\
   -0.03  &-0.4   & -1.3\emat
   $$
so that the eigenvalues of $A$ are circa $-0.89$,
$-0.11$ and
$-0.3$.
Let $Q=qq\tp,$ with 
$q:=[4\ \ 36\ \ 29]\tp$.
Then the solution $X$ of (\ref{lyap-red}) is
$$
X\simeq\bmat{ccc}
33.4  &  294.8  &  266.7\\
294.8  &  2607.9  &  2356.9\\
266.7  &  2356.9  &  2136.5\emat
$$
and it is easy to check that the determinant of the $2\times 2$ upper-left corner of $X$ is negative.

There is however, an interesting case in which the solution is indeed totally positive.

\begin{theorem}
Under the assumptions of Theorem \ref{main-res}, if 
$$
q=\sum_{j=1}^n \alpha_j v_j$$
where $v_j$ are given by (\ref{eigvj}) and $\alpha_j\geq 0$, 
then the solution $X$ of (\ref{lyap-red}) is totally positive.
\end{theorem}
\begin{proof}
We can represent the solution $X$ 
as
\beq\label{intrepX}
X=\int_{0}^{\infty}\exp(A\tp t)qq\tp\exp(A t)dt=
V^{-1} \int_{0}^{\infty}\exp(D t)Vqq\tp V\tp \exp(D t)dt\;V^{-\top}
\eeq
Now let 
$$w=[w_1\ w_2\ \dots\ w_n]\tp:=Vq=V \sum_{j=1}^n \alpha_j v_j=
\sum_{j=1}^n \alpha_j V  v_j$$
We observe that,in view of Equation (\ref{eigvinv}), $v_j$ are the columns of $V^{-1}$ multiplied by
$p_j:={\prod_{\substack{k=1 \\ k \neq j}}^n (\lambda_k - \lambda_j)}$
so that
$$
V  v_j={\prod_{\substack{k=1 \\ k \neq j}}^n (\lambda_k - \lambda_j)} e_j=p_je_j$$
where $e_j$ is the $j-$th canonical vector.
As a consequence
$
w=\sum_{j=1}^n \alpha_jp_je_j,$ so that
$$w_j= \alpha_j p_j.$$
Notice that, since the eigenvalues $\lambda_i$ are arranged in increasing order, ${\rm sign}(p_j)=(-1)^{j-1}$. Since $\alpha_j\geq 0$, this implies
$w_j\geq 0$ for all odd $j$, and $w_j\leq 0$ for all even $j$.
Hence, if we set $D_w:={\rm{diag}}(w)$ to be the diagonal matrix having the elements of $w$ in the main diagonal and $D_{\pm}:={\rm{diag}}([1,-1,1,\dots (-1)^{n-1}])$, the diagonal matrix 
$$D_0:=D_w D_{\pm}$$ 
has only nonnegative entries and hence is totally positive.

From (\ref{intrepX}), it follows
\beq\label{2intexX}
X=
V^{-1} \int_{0}^{\infty}\exp(D t)w w\tp \exp(D t)dt\;V^{-\top}.
\eeq

We can explicitly compute 
$$K:=\int_{0}^{\infty}\exp(D t)Vqq\tp V\tp \exp(D t)dt$$
whose entries turn out to be
$$
K_{i,j}=\frac{w_iw_j}{\lambda_i+\lambda_j}.$$
As a consequence $K=D_w C_0 D_w$, where $C_0$ is the Cauchy matrix associated with the positive increasing nodes $\lambda_i$. Its entries are
$$
[C_0]_{i,j}:=\frac{1}{\lambda_i+\lambda_j}.$$
We recall that Cauchy matrices associated with positive increasing nodes
are totally positive (indeed strictly totally positive), \cite[page 92]{Pinkus-TPM}.

Plugging $K=D_w C_0 D_w$ into (\ref{2intexX}), we get
$$
X=V^{-1}D_w C_0 D_wV^{-\top}.
$$

Now observe that
 $$V=V_+D_{\pm},$$
where 
$$
V_+:=\bmat{cccccc}
1&\lambda_1 & \lambda_1^2 &\dots & \lambda_1^{n-1}\\
1& \lambda_2 & \lambda_2^2 &\dots & \lambda_2^{n-1}\\
1& \lambda_3 & \lambda_3^2 &\dots & \lambda_3^{n-1}\\
\vdots&\vdots&\dots&\vdots&\vdots\\
1& \lambda_n & \lambda_n^2 &\dots & \lambda_n^{n-1}
\emat$$
The matrix  $V_+$ is a Vandermonde matrix associated with the positive and increasing nodes $\lambda_i$ so that it  is totally positive \cite[Section 4.2]{Pinkus-TPM}; moreover  $D_{\pm}=D_{\pm}^{-1}$.
Hence,
$$V^{-1}=D_{\pm}V_+^{-1}=D_{\pm}D_{\pm}P_+ D_{\pm}=P_+ D_{\pm},$$
 where $P_+:=D_{\pm}V_+^{-1} D_{\pm}$ is also totally positive
 \cite[Proposition 1.6]{Pinkus-TPM}.

In conclusion,
$$
X=P_+ D_{\pm} D_w C_0 D_w D_{\pm} P_+\tp = P_+ D_0 C_0 D_0 P_+\tp
$$ 
can be written as a product of totally positive matrices and hence it is itself  totally positive, \cite[Proposition 1.4]{Pinkus-TPM}
\end{proof}

\begin{remark}
Notice that this result is much stronger as total positivity is much more than positive semi-definiteness and entrywise nonnegativity. 
\end{remark}  

\section{Conclusion}

We have considered the Lyapunov equation
\[
XA+A\tp X=-Q,
\qquad Q=Q\tp\succeq0,
\]
where $A$ is a Hurwitz companion matrix.  While the positive
semidefiniteness of the unique solution $X$ follows from the classical
Lyapunov theory, the entrywise nonnegativity of $X$ is a separate and more
special property.  This property is relevant from the point of view of
positive systems, where the componentwise order carries direct physical or
probabilistic meaning.

The main result proves the desired entrywise nonnegativity when the companion
matrix has only real eigenvalues.  After reducing to rank-one right-hand
sides $Q=qq\tp$ and, by approximation, to the case of simple eigenvalues, the
forcing vector is expanded in the eigenbasis of $A\tp$.  Each entry of the
solution is then written as a quadratic form whose matrix has the Cauchy-like
structure (\ref{Cauchy-likemat}).
The positive semidefiniteness of this matrix, which may be viewed as dual to the class of Kwong matrices \cite{Kwong1989,BhatiaSano2009,BhatiaJain2023}, also known as anti-Loewner matrices, implies that every entry of \(X\) is nonnegative.

We also examined whether a stronger total-positivity property can be expected.
A simple example shows that the solution need not be totally nonnegative in
general.  Nevertheless, if the rank-one forcing vector has a nonnegative
expansion in the eigenvectors of $A\tp$, then the solution admits a
factorization into totally nonnegative factors, and hence is itself totally
nonnegative in the terminology used here.

Several questions remain open.  The most natural one is whether the real-root
assumption can be removed, thereby proving the original conjecture for all
Hurwitz companion matrices.  It would also be interesting to identify
additional conditions guaranteeing strict entrywise positivity or strict total
positivity, and to clarify whether analogous phenomena occur for other
structured realizations beyond companion form.

\bibliographystyle{IEEEtran}

\bibliography{Companion}

%
%
%
%
%

\end{document}